\documentclass[a4paper,11pt,accepted=2023-04-30]{quantumarticle}
\pdfoutput=1
\usepackage[utf8]{inputenc}
\usepackage[english]{babel}
\usepackage[T1]{fontenc}
\usepackage{amsmath}

\usepackage[numbers,sort&compress]{natbib}

\usepackage{tikz}
\usepackage{lipsum}

\usepackage{geometry}
\usepackage{amsthm}
\usepackage{amsmath}
\usepackage{amssymb}
\usepackage{url}

\makeatletter
\numberwithin{equation}{section}
\numberwithin{figure}{section}
\newtheorem{thm}{Theorem}
  \newtheorem{defn}{Definition}
  
  \newtheorem{fact}{Fact}

\@ifundefined{date}{}{\date{}}
\usepackage{tikz} 

\input{Qcircuit}
\usepackage{graphicx,amsmath, amsthm, amssymb,color,url,booktabs,comment}  
\usepackage[colorlinks]{hyperref} 

\newcommand{\nc}{\newcommand}
\nc{\rnc}{\renewcommand}
\newcommand{\proj}[1]{|#1\rangle\langle #1|}

\nc{\vev}[1]{\langle#1\rangle}
\nc{\grad}{{\vec{\nabla}}}

\DeclareMathOperator{\poly}{poly}
\DeclareMathOperator{\polylog}{polylog}
\DeclareMathOperator{\tr}{Tr}
\DeclareMathOperator{\rank}{rank}

\DeclareMathOperator{\supp}{supp}

\newcommand{\be}{\begin{equation}}
\newcommand{\ee}{\end{equation}}
\newcommand{\bea}{\begin{eqnarray}}
\newcommand{\eea}{\end{eqnarray}}
\newcommand{\nn}{\nonumber}
\newcommand{\bi}{\begin{itemize}}
\newcommand{\ei}{\end{itemize}}
\newcommand{\bn}{\begin{enumerate}}
\newcommand{\en}{\end{enumerate}}
\def\beas#1\eeas{\begin{eqnarray*}#1\end{eqnarray*}}
\def\ba#1\ea{\begin{align}#1\end{align}}
\nc{\bas}{\[\begin{aligned}}
\nc{\eas}{\end{aligned}\]}
\nc{\bpm}{\begin{pmatrix}}
\nc{\epm}{\end{pmatrix}}

\def\nn{\nonumber}

\def\L{\left} 
\def\R{\right}

\newtheorem{cor}[thm]{Corollary}

\makeatletter
\newtheorem*{rep@theorem}{\rep@title}
\newcommand{\newreptheorem}[2]{%
\newenvironment{rep#1}[1]{%
 \def\rep@title{#2 \ref{##1}}%
 \begin{rep@theorem}}%
 {\end{rep@theorem}}}
\makeatother

\newreptheorem{thm}{Theorem}
\newreptheorem{lem}{Lemma}

\def\eps{\epsilon}

\def\cE{\mathcal{E}}

\def\cH{\mathcal{H}}

\def\cL{{\cal L}}

\def\cX{{\cal X}}
\def\cY{{\cal Y}}

\DeclareMathOperator*{\E}{\mathbb{E}}

\def\benum{\begin{enumerate}}
\def\eenum{\end{enumerate}}
\def\bit{\begin{itemize}}
\def\eit{\end{itemize}}
\def\bdesc{\begin{description}}
\def\edesc{\end{description}}

\newcommand{\corref}[1]{Corollary~\ref{cor:#1}}

\nc{\todo}[1]{\textcolor{red}{todo: #1}}

\def\begsub#1#2\endsub{\begin{subequations}\label{eq:#1}\begin{align}#2\end{align}\end{subequations}}
\nc\qand{\qquad\text{and}\qquad}
\nc\mnb[1]{\medskip\noindent{\bf #1}}
\nc\mn{\medskip\noindent}

\nc{\nl}{\nn \\ &=}  
\nc{\nnl}{\nn \\ &}  
\nc{\fot}{\frac{1}{2}} 
\nc{\oo}[1]{\frac{1}{#1}} 
\newcommand{\ben}{\begin{enumerate}}
\newcommand{\een}{\end{enumerate}}
\nc{\mc}{\mathcal}
\nc{\beq}{\begin{equation}}
\nc{\eeq}{\end{equation}}
\nc{\norm}[1]{\L\| #1 \R\|}

\nc{\onenorm}[1]{\L\| #1 \R\|_1} 

\newcommand{\hannote}[1]{\textcolor{blue}{\small {\textbf{(Han:} #1\textbf{) }}}}

\nc{\Ra}{\Rightarrow}
\nc{\zo}{\{0,1\}}	
\nc{\trp}[1]{\tr \L( #1 \R)} 
\makeatother

\usepackage{babel}

\usepackage{mathtools}

\newcommand {\br} [1] {\ensuremath{ \left( #1 \right) }}

\newcommand {\minusspace} {\: \! \!}

\newcommand {\fn} [2] {\ensuremath{ #1 \minusspace \br{ #2 } }}

\newcommand{\sfrec}{\mathsf{rec}}
\newcommand{\sfvc}{\mathsf{VC}}
\newcommand{\sfcs}{\mathsf{CS}}

\newcommand{\sfD}{\mathsf{R}}
\newcommand{\Ext}{\mathsf{Ext}}
\newcommand {\Hmin}{\mathrm{H}_{\min}}
\newcommand{\sfR}{\mathsf{R}}
\newcommand{\sfQ}{\mathsf{Q}}
\newcommand{\Tr}{\mathrm{Tr}}
\newcommand\boddu[1]{\textcolor{red}{#1}}
\newcommand{\ketbra}[1]{|#1\rangle\langle#1|}
\newcommand{\defeq}{\ensuremath{ \stackrel{\mathrm{def}}{=} }}
\newcommand{\X}{\mathcal{X}}
\newcommand{\Y}{\mathcal{Y}}
\newcommand{\Z}{\mathcal{Z}}

\newcommand{\M}{\mathcal{M}}
\newcommand{\cP}{\mathcal{P}}

\newcommand{\I}{\mathrm{I}}
\newcommand{\mI}{\mathrm{I}_{\max}}
\newcommand{\id}{\mathbb{I}}

\newcommand{\err}{\mathrm{err}}
\newcommand{\var}{\mathsf{Var}}
\newcommand{\suppress}[1]{}
\newcommand{\dmax}[2]{\fn{\mathrm{D}_{\max}}{#1 \middle\| #2}}
\newcommand {\imax}{\ensuremath{\mathrm{I}_{\max}}}

\newcommand{\cc}{\mathrm{CC}}
\newcommand {\hmin} [2] {\fn{ \mathrm{H }_{\min}}{#1 \middle | #2}}

\title{On relating one-way classical and quantum communication complexities}

\author{Naresh Goud Boddu}
\email{naresh.boddu@ntt-research.com}
\affiliation{NTT Research, Sunnyvale, California, USA}
\orcid{0000-0001-6595-572X}
\author{Rahul Jain}
\email{rahul@comp.nus.edu.sg}
\affiliation{Centre for Quantum Technologies and Department of Computer Science, National University of Singapore and MajuLab, UMI 3654, Singapore}
\author{Han-Hsuan Lin}
\email{linhh@cs.nthu.edu.tw}
\orcid{0000-0002-5126-0174}
\affiliation{National Tsing Hua University, Taiwan}

\begin{document}

\maketitle

  \begin{abstract}
 Communication complexity is the amount of communication needed to compute a function when the function inputs are distributed over multiple parties. In its simplest form, one-way communication complexity, Alice and Bob compute a function $f(x,y)$, where $x$ is given to Alice and $y$ is given to Bob, and only one message from Alice to Bob is allowed. A fundamental question in quantum information is the relationship between one-way quantum and classical communication complexities, i.e., how much shorter the message can be if Alice is sending a quantum state instead of bit strings? We make some progress towards this question with the following results.
 
 

  Let $f: \X \times \Y \rightarrow \Z \cup \{\bot\}$ be a partial function and $\mu$ be a distribution with support contained in $f^{-1}(\Z)$. Denote $d=|\Z|$. Let $\sfD^{1,\mu}_\eps(f)$ be the classical one-way communication complexity of $f$; $\sfQ^{1,\mu}_\eps(f)$ be the  quantum one-way communication complexity of $f$ and $\sfQ^{1,\mu, *}_\eps(f)$ be the entanglement-assisted quantum one-way communication complexity of $f$, each with distributional error (average error over $\mu$) at most $\eps$. We show:
  \begin{enumerate}
      \item If $\mu$ is a product distribution, $\eta > 0$ and $0 \leq \eps \leq 1-1/d$, then,
      \begin{align*}
          &\sfR^{1,\mu}_{2\eps -d\eps^2/(d-1)+ \eta}(f) \\ & \quad \quad  \leq 2\sfQ^{1,\mu, *}_{\eps}(f)  +  O(\log\log (1/\eta))\enspace.
      \end{align*}
   In other words for $\delta, \eta > 0$ (by setting $\eps = 1 - \frac{1}{d} - \delta$),
    \begin{align*}
          &\sfR^{1,\mu}_{1- \frac{1}{d} - \frac{d}{d-1} \delta^2 + \eta}(f) \\ & \quad \quad \leq 2 \sfQ^{1,\mu, *}_{1- \frac{1}{d} - \delta}(f) +  O(\log\log (1/\eta))\enspace.
      \end{align*} We show similar results for other related communication models. 
   \item If $\mu$ is a non-product distribution and $\Z=\zo$, then $\forall \eps, \eta > 0$ such that $\eps/\eta + \eta < 0.5$,
   \[\sfD^{1,\mu}_{3\eta}(f) = O(\sfQ^{1,\mu}_{{\eps}}(f) \cdot \sfcs(f)/\eta^3)\enspace, \]where
      	\[\sfcs(f) = \max_{y} \min_{z\in\{0,1\}}  \vert \{x~|~f(x,y)=z\} \vert  \enspace.\]

  \end{enumerate}
  
\end{abstract}

\section{Introduction}Communication complexity concerns itself with characterizing the minimum number of bits or qubits that distributed parties need to exchange in order to accomplish a given task (such as computing a function $f$). Over the years, different models of communication for two party and multi party communication~\cite{BDHT99} have been proposed and studied. We consider only two party communication models in this paper. Communication complexity models have established striking connections with other areas in theoretical computer science, such as data structures, streaming algorithms, circuit lower bounds, decision tree complexity, VLSI designs, etc.

 In the two-way communication model, two parties Alice and Bob receive an input $x \in \X$ and $y \in \Y$ respectively. They interact with each other, communicating several messages, in order to jointly compute a given function $f(x,y)$ of their inputs. Their goal is to do this with as little communication as possible. Suppose if only one message is allowed, say from Alice to Bob, and Bob outputs $f(x,y)$ without any further interaction with Alice, then the model is called one-way. We refer readers to the textbook of Kushilevitz and Nisan~\cite{KN96} for a comprehensive introduction to the field of classical communication complexity. The work of Yao~\cite{Yao93} introduced quantum communication complexity, and since then various other analogous quantum communication models are proposed and studied. In the quantum communication models, the parties send quantum messages and  are allowed to use quantum operations. 

 In the current paper, we study the relation between quantum and classical one-way communication complexities. Let $\sfR^{1}_{\eps}(f)$ denote the classical one-way communication complexity of $f$~(Alice and Bob are allowed to use public and private  randomness independent of the inputs);  $\sfQ^{1}_{\eps}(f)$ denote the quantum one-way communication complexity of $f$ and $\sfQ^{1,*}_{\eps}(f)$ denote the entanglement-assisted quantum one-way communication complexity of $f$, each with worst case error $\epsilon$. Let $\mu$ be a probability distribution over $\X \times \Y$ and $\mu_X$ be the marginal of $\mu$ on $\X$. Let $\sfR^{1, \mu}_{\eps}(f)$ represent the classical one-way communication complexity of $f$; $\sfQ^{1, \mu}_{\eps}(f)$ denote the quantum one-way communication complexity of $f$ and $\sfQ^{1,\mu,*}_{\eps}(f)$ denote the entanglement-assisted quantum one-way communication complexity of $f$, each with distributional error (average error over $\mu$) at most $\eps$. Let $\sfR^{1, \mu_X}_{\eps}(f)$ represent the classical one-way communication complexity of $f$ with distributional error for worst case $y$ while $x$ is averaged over the distribution $\mu_X$ at most $\eps$; $\sfQ^{1, \mu_X}_{\eps}(f)$ denote the quantum one-way communication complexity of $f$ and $\sfQ^{1,\mu_X,*}_{\eps}(f)$ denote the entanglement-assisted quantum one-way communication complexity of $f$, each with distributional error for worst case $y$ while $x$ is averaged over the distribution $\mu_X$ at most $\eps$. Please refer to Section~\ref{sec:prelims}~[\ref{subsect:comm_complex}] for precise definitions.

 A fundamental question about one-way communication complexity is the relation between $\sfR^{1}_{\eps}(f)$ and $\sfQ^{1}_{\eps}(f)$ (or $\sfQ^{1,*}_{\eps}(f)$). Clearly 
 $\sfQ^{1,*}_{\eps}(f) \leq \min\{\sfQ^{1}_{\eps}(f), \sfR^{1}_{\eps}(f)\}$. When $f$ is a partial function, Gavinsky et al.\newline~\cite{GKKRW07} established an exponential separation between $\sfR^{1}_{\eps}(f)$ and $\sfQ^{1}_{\eps}(f)$. It is a long standing open problem to relate $\sfR^{1}_{\eps}(f)$ and $\sfQ^{1,*}_{\eps}(f)$ for a total function $f$. Since both measures are related to their distributional versions, $\sfD^{1, \mu}_{\eps}(f)$ and $\sfQ^{1, \mu, *}_{\eps}(f)$, via Yao's Lemma~\cite{Yao79}, we study the problem of relating measures $\sfD^{1, \mu}_{\eps}(f)$ and $\sfQ^{1, \mu, *}_{\eps}(f)$ for a fixed distribution $\mu$.

\subsection*{Previous results}

For a total function $f : \X \times \Y \to \{0, 1\}$, its Vapnik-Chervonenkis (VC) dimension, denoted by $\sfvc(f)$, is an important complexity measure, widely studied specially in the context of computational learning theory. If $\mu$ is a product distribution, Kremer, Nisan and Ron~\cite{KNR95} established a connection between the measures  $\sfD^{1, \mu}_{\eps}(f)$ and $\sfvc(f)$ as follows:
\begin{align}
    \sfD^{1, \mu}_{\eps}(f) = O\left(\frac{1}{\eps} \log\left( \frac{1}{\eps}\right)  \sfvc(f)\right). \nn
\end{align}
Ambainis et al.~\cite{ANTV99} showed the following:
\begin{align}
    \max_{\text{product } \lambda} \sfQ^{1, \lambda, *}_{\eps}(f) = \Omega \left( \sfvc(f)\right). \nn
\end{align}
Above equations establish that for a product distribution $\mu$, 
$$\max_{ \text{product } \lambda} \sfQ^{1, \lambda, *}_{\eps}(f) = \Omega(\sfD^{1, \mu}_{\eps}(f)).$$ 
Jain and Zhang~\cite{JZ09} extended the result of \cite{KNR95} when  $\mu$ is any  (non-product) distribution given as follows:
 \begin{align*}
   &\sfD^{1, \mu}_{\eps}(f) \\ &= O\left(\frac{1}{\eps} \log \left( \frac{1}{\eps}\right) \left( \frac{\I(X:Y)}{\eps}+1\right) \sfvc(f)\right). \nn
\end{align*}
For a function $f : \X \times \Y \to \{0, 1\}$, another measure that is often very useful in understanding classical one-way communication complexity, is the rectangle bound (denoted $\sfrec(f)$) a.k.a. the corruption bound. The rectangle bound $\sfrec(f)$ is  defined via a distributional version $\sfrec^{\mu}(f)$. It is a well-studied measure and $\sfrec^{1,\mu}(f)$ is well known to form a lower bound on $\sfD^{1,\mu}(f)$. If $\mu$ is a product distribution,~\cite{JZ09} showed,
\begin{align}
   \sfQ^{1, \mu}_{\eps^3}(f) = \Omega \left( \sfrec_{\eps}^{1, \mu}(f)\right). \nn
\end{align}
For a product distribution $\mu$, Jain, Klauck and Nayak~\cite{JKN08} showed,
\begin{align}
    \max_{ \text{product }\lambda} \sfrec_{\eps}^{1,\lambda}(f) = \Omega(\sfD^{1, \mu}_{\eps}(f)). \nn
\end{align}
Above equations establish that for a product distribution $\mu$, 
$$\max_{ \text{product } \lambda} \sfQ^{1, \lambda}_{\eps^3}(f) =\Omega(\sfD^{1, \mu}_{\eps}(f)).$$ 
However, it remained open whether $\sfD^{1, \mu}_{\eps}(f)$ and $\sfQ^{1, \mu}_{\eps}(f)$ (or $\sfQ^{1, \mu,*}_{\eps}(f)$) are related for a fixed distribution $\mu$. We answer it in positive and show the following results. 
\subsection*{Our results}
\begin{thm}\label{intro:thm1}
Let $f: \X \times \Y \rightarrow \Z \cup \{\bot\}$ be a partial function~\footnote{A partial function under a product $\mu$ is basically same as a total function.} and $\mu$ be a product distribution supported on $f^{-1}(\Z)$. Denote $d=|\Z|$. Let $\eta >0$ and $0 \leq \eps \leq1-1/d$.  Then,
\begin{align*}
&\sfR^{1,\mu}_{2\eps-d\eps^2/(d-1)+\eta}(f)  \\ &\quad \quad \leq  2\sfQ^{1,\mu,*}_{\eps}(f) +  O\L(\log\log (1/\eta)\R), \\
&\sfR^{1,\mu}_{2\eps-d\eps^2/(d-1)+\eta}(f)\\ &\quad \quad \leq\sfQ^{1,\mu}_{\eps}(f) +  O\L(\log\log (1/\eta)\R),\\
&\sfR^{1,\mu_X}_{2\eps-d\eps^2/(d-1)+\eta}(f) \\&\quad \quad \leq 2\sfQ^{1,\mu_X,*}_{\eps}(f)  +  O\L(\log\log (1/\eta)\R),\\
&\sfR^{1,\mu_X}_{2\eps-d\eps^2/(d-1)+\eta}(f) \\&\quad \quad  \leq\sfQ^{1,\mu_X}_{\eps}(f) +  O\L(\log\log (1/\eta)\R).
\end{align*}
\end{thm}
Note that for entanglement-assisted protocols, there must be a factor of 2 because of super dense coding\suppress{Note our bound has no loss in epsilon when $\eps=0$ or $\eps=1-1/d$.}. Additionally, if $\mu$ is a non-product distribution, we show, 
   \begin{thm}\label{intro:thm2}
Let $\eps, \eta >0$ be such that $\eps/\eta + \eta < 0.5$. Let $f: \X \times \Y \rightarrow \{0,1,\bot \}$ be a partial function and $\mu$ be a distribution supported on $f^{-1}(0) \cup f^{-1}(1)$. Then,$$\sfD^{1,\mu}_{3\eta}(f)=O\left(  \frac{\sfcs(f) }{\eta^3}  \sfQ^{1,\mu}_{{\eps}} (f)  \right),$$ 
      where
      	\[\sfcs(f) = \max_{y} \min_{z\in\{0,1\}} \{ \vert \{x~|~f(x,y)=z\} \vert\}  \enspace.\] 
\end{thm}
Both Theorem~\ref{intro:thm1} and Theorem~\ref{intro:thm2} are proved by converting quantum protocols into classical protocols directly.

The bound provided by Theorem~\ref{intro:thm2} depends on the column sparsity $\sfcs(f)$. Although $\sfcs(f)$ can be as large as $O(|\X|)$, giving a bound exponentially worse than the $O(\log(|\X|))$ brute force protocol, Theorem~\ref{intro:thm2} is useful when $\sfcs(f)$ is constant. In particular, Theorem~\ref{intro:thm2} can convert the quantum fingerprinting protocol~\cite{buhrman2001quantum}~\footnote{This protocol is proposed for simultaneous message passing model, but it can be easily converted into one for one-way communication model.} on EQUALITY function into a classical communication protocol with similar complexity for the worst case by combining it with Yao's Lemma~\cite{Yao79}.
 
   \subsection*{Proof overview}


For a product distribution $\mu$, we upper bound $\sfD^{1, \mu}_{\eps}(f)$ by $\sfQ^{1, \mu,*}_{\eps}(f)$, using ideas from K{\"o}nig and Tehral~\cite{konig2008bounded} and 
 Jain,  Radhakrishnan, and  Sen~\cite{JRS03,JRS05}.
 For an entanglement-assisted quantum one-way communication protocol, let $Q\equiv DE_B$ represent Alice's quantum message $D$ and Bob's part of the entanglement $E_B$. We first replace Bob's measurement by the pretty good measurement (PGM) (with a small loss in the error probability). Then we use an  idea of~\cite{konig2008bounded} to show that we can "split" Bob's PGM into the PGM for guessing $X$. 
Since this new $X$-guessing PGM is independent of $Y$, Alice can apply it herself on the register  $Q$ (Alice's message and Bob's share of prior entanglement) and send the measurement outcome $C$ to Bob, who will just output $f(C,Y)$. The classical message that Alice sent is long (in fact it is equal to the length of $X$) but it has low max-information with input $X$, since (by monotonicity of the max-information) $\imax(X:C) \leq \imax(X:Q) \leq  2\log( \vert  D \vert )$. Note that the second inequality has a factor of $2$ due to super dense coding. We then use a compression protocol from~\cite{JRS03,JRS05} to compress $C$ into another short message $C'$ of size $2\log( \vert  D \vert )$. The same argument works for variants of this result where the two parties does not share entanglement, and/or where the error probability is averaged over a distribution of $x$ and maximized over $y$.

For a non-product distribution $\mu$, we upper bound $\sfD^{1, \mu}_{\eps}(f)$ by $\sfQ^{1, \mu}_{\eps}(f)$, using ideas of Huang,  Kueng and Preskill~\cite{HK19} and~\cite{JRS03,JRS05}. For a quantum one-way communication protocol with quantum message~\footnote{We assume, by at most doubling the message size that Alice's message for any input is a pure state.} $Q$, we first use the idea of~\cite{HK19} to show that there exists a "classical shadow" $C$ of the quantum message $Q$, which will allow Bob to estimate $\Tr(E^y_bQ)$ (for any $b \in \{0,1\}$, where $\M^y=\{ E^y_0, E^y_1\}$ is Bob's measurement on input $y$). This allows Alice to send the classical shadow $C$ of quantum message $Q$. However, the precision of the  classical shadow procedure of~\cite{HK19} depends on $\Vert E^y_b \Vert_F^2$, so we need to bound  $\Vert E^y_b \Vert_F^2$.
We show that there exists measurement operator $\tilde{E}^y_b$ (for some $b \in \{0,1\}$) such that $\Vert \tilde{E}^y_b \Vert_F^2$ is at most the "column sparsity" of function $f$ and $\Tr(\tilde{E}^y_bQ)$ is "close" to $\Tr(E^y_bQ)$.
We again note that the classical shadow has low max-information with input $X$, since (using the monotonicity for max-information) $\imax(X:C) \leq \imax(X:Q) \leq  \log( \vert  Q \vert )$. As before, we use the compression protocol from~\cite{JRS03,JRS05} to compress $C$ into another short message $C'$ of size $\log( \vert  Q \vert )$.

\subsection*{Organization} 
In Section~\ref{sec:prelims}, we present our notations, definitions and other information theoretic preliminaries. In Section~\ref{sec3}, we present the proof of Theorem~\ref{intro:thm1}.  In Section~\ref{sec4}, we present the proof of Theorem~\ref{intro:thm2}.

\section{Preliminary}

\label{sec:prelims}
\subsection*{Quantum information theory} All the logarithms are evaluated to the base $2$. Consider a finite dimensional Hilbert space $\cH$ endowed with an inner-product $\langle \cdot, \cdot \rangle$ (we only consider finite dimensional Hilbert spaces). A quantum state (or a density matrix of a state) is a positive semi-definite matrix on $\cH$ with trace equal to $1$. It is called {\em pure} if and only if its rank is $1$.  Let $\ket{\psi}$ be a unit vector on $\cH$, that is $\langle \psi,\psi \rangle=1$.  With some abuse of notation, we use $\psi$ to represent the state and also the density matrix $\ketbra{\psi}$, associated with $\ket{\psi}$. Given a quantum state $\rho$ on $\cH$, {\em support of $\rho$}, called $\text{supp}(\rho)$ is the subspace of $\cH$ spanned by all eigenvectors of $\rho$ with non-zero eigenvalues.

A {\em quantum register} $A$ is associated with some Hilbert space $\cH_A$. Define $\vert A \vert \defeq \dim(\cH_A)$ and $\ell(A) = \log \vert A \vert$. Let $\mathcal{L}(\cH_A)$ represent the set of all linear operators on $\cH_A$ and $\mathcal{D}(\cH_A)$, the set of all quantum states on $\cH_A$. For operators $O, O'\in \cL(\cH_A)$, the notation $O \leq O'$ represents the L\"{o}wner order, that is, $O'-O$ is a positive semi-definite matrix. State $\rho$ with subscript $A$ indicates $\rho_A \in \mathcal{D}(\cH_A)$. If two registers $A,B$ are associated with the same Hilbert space, we shall represent the relation by $A\equiv B$. For two states $\rho_A, \sigma_B$, we let $\rho_A \equiv \sigma_B$ represent that they are identical as states, just in different registers. Composition of two registers $A$ and $B$, denoted $AB$, is associated with the Hilbert space $\cH_A \otimes \cH_B$.  For two quantum states $\rho\in \mathcal{D}(\cH_A)$ and $\sigma\in \mathcal{D}(\cH_B)$, $\rho\otimes\sigma \in \mathcal{D}(\cH_{AB})$ represents the tensor product ({\em Kronecker} product) of $\rho$ and $\sigma$. The identity operator on $\cH_A$ is denoted by $\id_A$. Let $U_A$ denote maximally mixed state in $\cH_A$. Let $\rho_{AB} \in \mathcal{D}(\cH_{AB})$. Define
$$ \rho_{B} \defeq \tr_{A}{\rho_{AB}} \defeq \sum_i (\bra{i} \otimes \id_{B})
\rho_{AB} (\ket{i} \otimes \id_{B}) , $$
where $\{\ket{i}\}_i$ is an orthonormal basis for the Hilbert space $\cH_A$.
The state $\rho_B\in \mathcal{D}(\cH_B)$ is referred to as the marginal state of $\rho_{AB}$. Unless otherwise stated, a missing register from subscript in a state will represent partial trace over that register. Given $\rho_A\in\mathcal{D}(\cH_A)$, a {\em purification} of $\rho_A$ is a pure state $\rho_{AB}\in \mathcal{D}(\cH_{AB})$ such that $\tr_{B}{\rho_{AB}}=\rho_A$. Purification of a quantum state is not unique. Suppose $A\equiv B$. Given $\{\ket{i}_A\}$ and $\{\ket{i}_B\}$ as orthonormal bases over $\cH_A$ and $\cH_B$ respectively, the \textit{canonical purification} of a quantum state $\rho_A$ is $\ket{\rho_A} \defeq (\rho_A^{\frac{1}{2}}\otimes\id_B)\left(\sum_i\ket{i}_A\ket{i}_B\right)$. Note that the size (number of qubits) of the  canonical purification $\ket{\rho_A} $ is twice the size of quantum state $\rho_A$.

A quantum channel $\cE: \mathcal{L}(\cH_A)\rightarrow \mathcal{L}(\cH_B)$ is a completely positive and trace preserving (CPTP) linear map (mapping states in $\mathcal{D}(\cH_A)$ to states in $\mathcal{D}(\cH_B)$). A {\em unitary} operator $U_A:\cH_A \rightarrow \cH_A$ is such that $U_A^{\dagger}U_A = U_A U_A^{\dagger} = \id_A$. The set of all unitary operators on $\cH_A$ is  denoted by $\mathcal{U}(\cH_A)$. An {\em isometry}  $V:\cH_A \rightarrow \cH_B$ is such that $V^{\dagger}V = \id_A$ and $VV^{\dagger} = \id_B$. A POVM element is an operator $0 \le M \le \id$. We use shorthand $\bar{M} \defeq \id - M$, where $\id$ is clear from the context. We use shorthand $M$ to represent $M \otimes \id$, where $\id$ is clear from the context. A measurement $\M = \{M_1, M_2, \ldots, M_t\}$ (with POVM elements $\{M^\dagger_1 M_1, M^\dagger_2 M_2, \ldots, M^\dagger_t M_t\}$) is a set of operators such that $\sum_{i=1}^{t}M^\dagger_i M_i = \id$. When $\M$ is performed on a state $\rho$, we get as outcome a random variable $\M(\rho)$, such that $\Pr(\M(\rho) = i)=\Tr(M_i \rho M^\dagger_i )$ and the state conditioned on outcome $i$ is $\frac{M_i \rho M^\dagger_i}{\Tr(M_i \rho M^\dagger_i)}$. A projector $\Pi$ is an operator such that $\Pi^2 =\Pi$, i.e. its eigenvalues are either $0$ or $1$.

For a classical random variable $X$, we use $x \leftarrow X$ to denote $x$ is drawn from the  distribution $P_X(x) \defeq \Pr(X=x)$. A {\em classical-quantum state} (cq-state) $\rho_{XQ}$ (with $X$ a classical random variable) is of the form \[ \rho_{XQ} =  \sum_{x \in \cX}  P_X(x)\ket{x}\bra{x} \otimes \rho^x_Q , \] where ${\rho^x_Q}$ are states and $P_X(x) =\Pr(X=x)_\rho$. For an event $G \subseteq \cX =\supp(X)$, define  
\[\Pr(G)_\rho =  \sum_{x \in G} P_X(x) \quad ; \quad  \]
\[ (\rho|G)\defeq \frac{1}{\Pr(G)_\rho} \sum_{x \in G} P_X(x)\ket{x}\bra{x} \otimes \rho^x_Q.\]
For a function $Z:\X \rightarrow \Z$, define \[ \rho_{ZXQ} \defeq  \sum_{x\in \cX}  P_{X}(x) \ket{Z(x)}\bra{Z(x)} \otimes \ket{x}\bra{x} \otimes  \rho^{x}_Q .  \]We also use $U_d$ to represent the uniform distribution over $\{0,1 \}^d$. 
\begin{defn}\label{def:infoquant}    
	\begin{enumerate}
		\item For $p \geq 1$ and matrix $A$,  let $\| A \|_p$ denote the {\em Schatten} $p$-norm. $\| A \|_2$ is also referred to as the Frobenius norm, denoted $\| A \|_F$.
		\item Let $\Delta(\rho , \sigma) \defeq \frac{1}{2} \|\rho - \sigma\|_1$.
		We write $\approx_\eps$ to denote $\Delta(\rho, \sigma) \le \eps$. 
		
	\item  For a quantum state $\rho$, and integer $t > 0$, we define 
	$$\rho^{\otimes t} \defeq \rho \otimes \rho \otimes \cdots \otimes \rho  \quad (\text{t times}).$$

	

	\end{enumerate}
\end{defn}

We start with the following fundamental information theoretic quantities. We refer the reader to the excellent sources for quantum information theory \cite{Wil12, Wat16} from where the facts stated below can be found.

\begin{defn}[von Neumann entropy]\label{def:entropy}
	 The von Neumann entropy of a quantum state $\rho$ is defined as, 
	  $$ \mathrm{S}(\rho) \defeq - \Tr(\rho\log\rho).$$
\end{defn}

\begin{defn}[Relative entropy]\label{def:relentropy}
Let $\rho, \sigma$ be states with  $\supp(\rho) \subset \supp(\sigma)$. The relative entropy between $\rho$ and $\sigma$ is defined as,
$$  \mathrm{D}(\rho \vert \vert \sigma) \defeq  \Tr(\rho\log\rho) - \Tr(\rho\log\sigma).$$
	
\end{defn}

\begin{defn}[Max-relative entropy~\cite{Datta_2009,JRS03}] Let $\rho, \sigma$ be states with  $\supp(\rho) \subset \supp(\sigma)$. The max-relative entropy between $\rho$ and $\sigma$ is defined as, 
$$ \dmax{\rho}{\sigma}  \defeq  \inf \{ \lambda \in \mathbb{R} : \rho \leq 2^{\lambda} \sigma  \}.$$  
\end{defn}

\begin{defn}[Max-information~\cite{Datta_2009}]\label{def:maxinfo}
  For state $\rho_{AB}$, $$ \imax(A:B)_{\rho} \defeq   \inf_{\sigma_{B}\in \mathcal{D}(\cH_B)}\dmax{\rho_{AB}}{\rho_{A}\otimes\sigma_{B}} .$$
  If $\rho$ is a classical state (diagonal in the computational basis) then the $\inf$ above is achieved by a classical state $\sigma_B$.
\end{defn}
\begin{defn}[Mutual information]
	\label{def:mutinfo}
	Let  $\rho_{ABC}$ be a quantum state. We define the following measures.
\begin{align*}
   & \text{Mutual information} : \quad \I(A:B)_{\rho} \\
    &\defeq  \mathrm{S}(\rho_A) +  \mathrm{S}(	\rho_B)
- \mathrm{S}(\rho_{AB}) \\ &=  \mathrm{D}(\rho_{AB} \vert \vert \rho_A\otimes\rho_B )  .
\end{align*}
 \begin{align*}
     & \text{Conditional mutual information}: \I(A:B~|~C)_{\rho}  \\
     & \defeq \I(A:BC)_{\rho}-\I(A:C)_{\rho}.
 \end{align*}
\end{defn}

\begin{fact}
\label{rhoablessthanrhoaidentity}
For a cq-state  $\rho_{XA}$ ($X$ classical): $\rho_{XA}  \le \id_X \otimes \rho_{A}$ and   $\rho_{XA}  \le \rho_X \otimes \id_{A}$.
\end{fact}
\begin{proof}
    For the first inequality consider,
    \begin{align*}
        \rho_{XA} &= \sum_x p_x \ketbra{x} \otimes \rho^x_A \\
        &\leq \sum_x \ketbra{x} \otimes \rho_A \\
        &= \id_X \otimes \rho_{A}.
    \end{align*}
For above note that $\sum_x p_x \rho^x_A = \rho_A$ and hence for all $x: p_x \rho^x_A \leq \rho_A$.

For the second inequality consider,
   \begin{align*}
        \rho_{XA} &= \sum_x p_x \ketbra{x} \otimes \rho^x_A \\
        &\leq \sum_x p_x \ketbra{x} \otimes \id_A \\
        &= \rho_X \otimes \id_{A}.
    \end{align*}
\end{proof}
\suppress{
\begin{defn}[Measurement result]
For a measurement $M$ with POVM elements $\M=\{M_i\}$ and a quantum state $\rho$, we define the measurement result $M(\rho)$ as a random variable whose value is $i$ with probability $\tr(M_i \rho)$, i.e. $ P_{M(\rho)}(i)=\Pr(M(\rho) =i) =\tr(M_i \rho) $.
\end{defn}
}
\begin{fact}[Monotonicity]\label{fact:mono}
	Let $\rho_{XA}$ be a cq-state ($X$ classical) and $\cE: \mathcal{L}(\cH_A)\rightarrow \mathcal{L}(\cH_B)$ be a CPTP map. Then,
	$$\mI(X:B)_{\cE(\rho)} \leq  \mI(X:A)_\rho  \leq \log{|A|}.$$
\end{fact}
\begin{proof}
    Let $c = \mI(X:A)_\rho$ and $\sigma_A$ be a state such that $\rho_{XA} \leq 2^c \rho_X \otimes \sigma_A$. Since $\cE$ preserves positivity,
    $(\id_X \otimes \cE)(\rho_{XA}) \leq 2^c \rho_X \otimes \cE(\sigma_A)$ which gives the first inequality.
    
    From Fact~\ref{rhoablessthanrhoaidentity}, 
    $\rho_{XA} \le \rho_X \otimes \id_A=2^{\log|A|}\rho_X \otimes U_A$.  The second inequality now follows from the definition of $\mI$.

\end{proof}
\suppress{
\begin{fact}\label{fact:bound}
 Let $\rho_{XBD}$ be a cq-state ($XD$ classical) such that $\rho_{XB} = \rho_X \otimes \rho_B$. Then, $$ \mI(X:BD)_\rho \leq  \log{(|D|)}.$$
\end{fact}
\begin{proof}
    Using Fact~\ref{rhoablessthanrhoaidentity}, we have $\rho_{XBD} \leq \rho_{XB} \otimes \id_D$. Since, $\rho_{XB} =\rho_X \otimes \rho_B$, we further have 
    $$\rho_{XBD} \leq  \rho_{X}\otimes \rho_{B} \otimes \id_D=2^{\log (\vert D \vert)} \rho_{X} \otimes \rho_B \otimes  U_D.$$ 
    Thus,
    $$\dmax{\rho_{XBD}}{\rho_{X}\otimes \rho_B \otimes U_{D}} \leq \log (\vert D \vert).$$ Noting,  $\imax(X : BD)_\rho = \inf_{\sigma_{BD}} \dmax{\rho_{XBD}}{\rho_{X}\otimes\sigma_{BD}}  \leq \dmax{\rho_{XBD}}{\rho_{X}\otimes \rho_B \otimes U_{D}}$, we have the desired  $\imax(X : BD)_\rho \leq \log (\vert D \vert).$
\end{proof}}
\begin{fact}[Naimark's theorem]\label{fact:naimark}
For a measurement $\M=\{M_1, M_2, \ldots, M_t\}$ and a quantum state $\rho_A$, there exists a unitary $U:\cH_{AZ} \rightarrow \cH_{AZ}$ such that $\vert Z \vert =t,$ and $ \Tr(M_i (\rho_A \otimes \ketbra{0}) M^\dagger_i)  = \Tr( (\id \otimes \ketbra{i} ) (U (\rho_A \otimes \ketbra{0}) U^\dagger) )  =  \Pr(Z=i)_{U (\rho_A \otimes \ketbra{0}) U^\dagger}$, for every $i \in [t].$
\end{fact}

\begin{fact}[Lemma B.7 in~\cite{berta2011quantum}]\label{fact:imax-dim-bound} For a state  $\rho_{XY}$, $\imax(X:Y)_\rho \leq 2 \log \min(|X|,|Y|)$.     
\end{fact}

\begin{fact}\label{fact:boundnew}
    Let $\rho_{XBD}$ be a cq-state ($X$ classical) such that $\rho_{XB} = \rho_X \otimes \rho_B$. Then, $$ \mI(X:BD)_\rho \leq  2\log{(|D|)}.$$
\end{fact} 
\begin{proof}
    From Fact~\ref{fact:imax-dim-bound}, we have  $\imax(XB : D)_\rho \leq 2 \log(\vert D \vert)$. Using Definition~\ref{def:maxinfo}, there exist a $\sigma_D$ such that 
    $$\rho_{XBD} \leq 2^{2  \log(\vert D \vert)} \rho_{XB} \otimes \sigma_D.$$
    Since $\rho_{XB} = \rho_X \otimes \rho_B$, we have
    $$\rho_{XBD} \leq 2^{2 \log(\vert D \vert)} \rho_{X} \otimes (\rho_B \otimes \sigma_D).$$
    Defining $\sigma_{BD}=\rho_B \otimes \sigma_D$, we have
    $$\rho_{XBD} \leq 2^{2 \log(\vert D \vert)} \rho_{X} \otimes (\sigma_{BD}).$$ From Definition~\ref{def:maxinfo}, we have the desired, $ \mI(X:BD)_\rho \leq  2\log{(|D|)}.$
\end{proof}

%
%





\begin{defn}[Projector on Hilbert space]\label{def:ProjHilbert}

Let $\cH$ be a Hilbert space with a basis $\{ v_i\}$. The projector on $\cH$ is defined as:   \[\mathsf{Proj}(\cH) \defeq \sum_{i} \ketbra{v_i}. \] 
		
\end{defn}
\begin{defn}[Guessing probability]Given a cq-state, $\rho_{XQ}=\sum_x p_x \proj{x} \otimes \rho^x_Q $, we often want to guess $X$ by doing a measurement on the quantum register $Q$. If we do so by a measurement $\M$ with POVM elements $\{E_x\}$, its success probability averaged over $X$ is \[\Pr[X=\M(Q)]= \sum_x p_x \trp{E_x \rho^x_Q}. \]We use $p^{opt}_g(X | Q)_\rho$ to denote the maximum probability over all measurements $\M$, i.e.
$$p^{opt}_g(X|Q)_\rho \defeq \max_{\M} \{\Pr[X=\M(Q)]\}.$$
\end{defn}

\begin{defn}[Pretty good measurement~(PGM)]
	
		For a cq-state, 
	$\rho_{XQ}=\sum_x p_x \proj{x} \otimes \rho^x_Q,$ define 
\ba
A_x=p_x \rho^x_Q, \quad \, A=\sum_x A_x. \nn
\ea
The pretty good measurement (PGM)  is the measurement  $\M^{pgm}_X$ with POVM elements $\{E^{pgm}_{x} = A^{-1/2} A_xA^{-1/2}\}$. We denote
\begin{align*}
    & p^{pgm}_{g}(X|Q)_\rho\\
    & \defeq  \sum_x p_x\trp{E^{pgm}_x \rho^x_Q} \\
    & =\Pr[X=\M^{pgm}_X(Q)]. 
\end{align*}
\end{defn}
 
%
%
%
%

\begin{fact}[Optimality of PGM~\cite{renes2017better}]\label{thm:pgm-opt}
	For any cq-state  
$\rho_{XQ}=\sum_x p_x \proj{x} \otimes \rho^x_Q,$ we have 
$$  p^{pgm}_{g}(X|Q)_\rho  \geq g(p^{opt}_{g}(X|Q)_\rho),$$where $g(x)=x^2+(1-x)^2/(d-1)$ and $d$ is the dimension of the register $X$.
\end{fact}

Note that $g(x)$ is convex everywhere and
increasing when $x\in [1/d,1]$, and $g(1/d)=1/d$. Also, note the bound from Fact~\ref{thm:pgm-opt} is better than the optimality bound of Barnum and Knill~\cite{pgm-bk} when the guessing probability is close to $1/d$.

%
\subsection*{One-way  communication complexity} \label{subsect:comm_complex}
In this paper we only consider the two party one-way model of communication. Let $f: \X \times \Y \to  \Z \cup \{ \bot \} $ be a partial function, $\mu$  be a distribution on $f^{-1}(\Z)$ and $\epsilon \geq 0$. Let $\mu_X$ represent the marginal of $\mu$ on $\X$. In a two party one-way communication protocol $\cP$, Alice with input $x \in \X$ communicates a message to Bob with input $y \in \Y$. On receiving Alice's message, Bob produces output of the protocol $\cP(x,y)$. 

In a one-way classical communication protocol, Alice and Bob are allowed to use public and private  randomness (independent of the inputs).\suppress{ By default all classical protocols considered are public-coin protocols.} In a one-way quantum communication protocol, Alice and Bob are allowed to do quantum operations and Alice can send a quantum message (qubits) to Bob. In an entanglement-assisted protocol, Alice and Bob start with a shared pure state (independent of the inputs) and Alice communicates a quantum message to Bob. 

 Let $\cP$ represent a one-way communication protocol.
\begin{defn}
    \begin{enumerate}
        \item $\err_{x,y}(\cP,f) \defeq \Pr(\cP(x,y) \neq f(x,y)),$~\footnote{For $(x,y) \notin f^{-1}(\Z) : \err_{x,y}(\cP,f)=0$.}
        \item $\err(\cP,f) \defeq \max_{x,y}\{ \err_{x,y}(\cP,f)\},$
        \item $\err(\cP,f,\mu) \defeq \E_{(x,y) \leftarrow \mu}[\err_{x,y}(\cP,f)],$
        \item $ \err(\cP,f,\mu_X) \defeq \max_y\{\E_{x \leftarrow \mu_X} [\err_{x,y}(\cP,f)]\}. $
        \item $\cc(\cP)$ be the maximum number of (qu)bits communicated in $\cP$. 
    \end{enumerate}
\end{defn}
\begin{defn} Let $\cP$ represent a classical public-coin protocol.
    \begin{enumerate}
        \item  $\sfR^{1}_{\eps}(f) \defeq \min\{\cc(\cP)~|~ \err(\cP,f)\le \eps\}$.
        \item  $\sfR^{1, \mu}_{\eps}(f) \defeq \min\{\cc(\cP)~|~\err(\cP,f,\mu) \le \eps\}$.
\item $\sfR^{1, \mu_X}_{\eps}(f)\defeq \min\{\cc(\cP)~|~\err(\cP,f,\mu_X) \le \eps\}$.
    \end{enumerate}
\end{defn}

Intuitively, $\sfR^{1}_{\eps}(f)$ is the classical communication complexity for worse case $(x,y)$, $\sfR^{1, \mu}_{\eps}(f)$ is the classical communication complexity with $(x,y)$ averaged over the distribution $\mu$, and
$\sfR^{1, \mu_X}_{\eps}(f)$ is the classical communication complexity with worst case $y$ while $x$ is averaged over the distribution $\mu$.
%
%

\begin{defn}
    \begin{enumerate}
        \item  The quantum one-way communication complexities $\sfQ^{1}_{\eps}(f)$, $\sfQ^{1, \mu}_{\eps}(f)$ and  $\sfQ^{1, \mu_X}_{\eps}(f)$ are defined similarly, by considering $\cP$ to be a quantum one-way protocol.   
        \item The entanglement-assisted quantum one-way communication complexities  $\sfQ^{1,*}_{\eps}(f)$, $\sfQ^{1,\mu,*}_{\eps}(f)$, and $\sfQ^{1, \mu_X,*}_{\eps}(f)$ are defined similarly, by considering $\cP$ to be an entanglement-assisted quantum one-way protocol.  

    \end{enumerate}
\end{defn}
The following result due to Yao~\cite{Yao77} is a very useful fact connecting worst-case and distributional communication complexities.
\begin{fact}[Yao's Principle~\cite{Yao77}]\label{yaofact1}
\[\sfR^{1}_{\eps}(f) = \max_{\mu} \sfR^{1, \mu}_{\eps}(f) \quad ; \quad \]
\[ \sfQ^{1,*}_{\eps}(f) = \max_{\mu} \sfQ^{1, \mu,*}_{\eps}(f). \]
\end{fact}

\begin{defn}[Column sparsity of a partial function]\label{def:sparsity} For a partial function $f: \X \times \Y \rightarrow \{0,1,\bot \}$, and for every $y \in \Y$, let $S^y_0= \{x|f(x,y)=0\}$,  $S^y_1= \{x|f(x,y)=1\}$. We define column sparsity ($\sfcs(f)$) of a partial function $f$ as follows:
	\[\sfcs(f) = \max_{y} \min \{ \vert S^y_0 \vert, \vert S^y_1 \vert \}.\]

\end{defn}

\begin{defn}[Markov-chain]
	\label{def:mc}
	Let  $ABC$ be joint random variables. We say $ABC$ forms a Markov-chain iff $\I(A:C~|~B)=0$ and denote it by $A \leftrightarrow B \leftrightarrow  C$.
\end{defn}
\begin{fact}[Message-compression~\cite{JRS03,JRS05}]
\label{fact:compress} Let $XC$ be joint random variables and $\eta>0$. Define,
\begin{align*}
    &T_\eta(X:C) \\
    &\defeq \min_{R,M} \{\ell(M) ~|~ X \leftrightarrow (R,M) \leftrightarrow C'~;~ \\ 
  & \quad XR = X \otimes R~;~ XC \approx_\eta XC'\} .
\end{align*}
Above $\ell(M)$ represents the length (number of bits) of $M$. Let $(M,R)$ achieve the minimum above. This means that if Alice and Bob share the public random string $R$, Alice can, with input $X$, generate $M$ (using $R,X$) and send $M$ to Bob, who in turn can produce $C'$ (using $(M,R)$). The communication from Alice to Bob is $T_\eta(X:C)$. 
Furthermore,
$$T_\eta(X:C)\leq \mI(X:C) +  O(\log\log (1/\eta)).$$
\end{fact}

\suppress{
\begin{fact}[Message-compression~\cite{harsha2007communication}]\label{thm:mutual-info-average} Let $XC$ be joint random variables. Define,
$$ T(X:C) \defeq \min_{R,D} \{\ell(D) ~|~ X \leftrightarrow (R,D) \leftrightarrow C~;~ XR = X \otimes R\}.$$
Above $\ell(D)$ represents expected length (number of bits) of $D$. Let $(D,R)$ achieve the minimum above. This means that if Alice and Bob share the public random string $R$, Alice can, with input $X$, generate $D$ (using $R,X$) and send $D$ to Bob, who in turn can produce $C$ (using $(D,R)$). The expected communication from Alice to Bob is $T(X:C)$. 
Furthermore,
$$\I(X:C)\leq T(X:C)\leq \I(X:C)+2\log(\I(X:C)+1)+O(1).$$
\end{fact}
}

%
%
%

\begin{fact}[Markov's inequality] \label{fact:markov}
For any nonnegative random variable $X$ and real number $a>0$, 
$$\Pr[X\geq a]  \leq \frac{\E[X]}{a}. $$
\end{fact}



\begin{fact}[]\label{factprojfrobrank}
	For a projector $\Pi$, we have
	$$ \Vert \Pi \Vert^2_F =\rank(\Pi).$$
\end{fact}

The following fact follows from Theorem $4$ in~\cite{HK19}.
\begin{fact}[Classical shadow~\cite{HK19}]\label{hk_trace}Fix $\epsilon, \delta \in (0,1)$ and $a>0$. Let $\rho \in \mathcal{D}(\cH_R)$ be a quantum state on $n$ qubits and $\vert R \vert= 2^n$. Let $T = O\L( a \frac{\log (\frac{ 1}{\delta})}{\eps^2}\R)$. Let $\M^{STAB}$ be the random stabilizer measurement, i.e. do a random Clifford unitary then measure in the computational basis. Let 
\begin{align*}
    &S \defeq  \M^{STAB}(\rho) \otimes \M^{STAB}(\rho) \\
    &  \quad \quad \quad  \otimes \ldots \otimes \M^{STAB}(\rho) \ \text{(T times)},
\end{align*}
where $\M^{STAB}(\rho)$ is a classical representation of the post-measurement stabilizer state. There exists a deterministic procedure $d(.)$ such that for any Hermitian matrix $A \in  \mathcal{L}(\cH_R)$ with $\Vert A \Vert^2_F \leq a$
$$\Pr_{s \leftarrow S} ( \vert d(A,s) - \Tr(A \rho) \vert \leq \eps )    \geq 1- \delta.$$
Additionally, $ \log ( \supp(S) ) \leq O(Tn^2)$. We call $S$ a classical shadow of $\rho$. 
\end{fact}


\begin{proof}
We provide a sketch of the proof for completion. Since there are $2^{O({n^2})}$ stabilizer states in $n$ qubits, $\M^{STAB}(\rho)$ has an efficient classical representation of $O(n^2)$ bits, and thus $S$ can be represented in $O(Tn^2)$ bits. This proves that $\log (\supp(S)) \leq O(Tn^2)$. Let $\ket{x}$ be the (random) stabilizer state corresponding to $\M^{STAB}(\rho)$. For $x \leftarrow \M^{STAB}(\rho)$, we define $d'(A,x) \defeq \Tr(A ((2^n+1)\ket{x}\bra{x} - \id))$. It can be shown that for any fixed Hermitian operator $A$,
\[ \E(d'(A,X))= \E_{x \leftarrow X}(d'(A,x)) = \Tr(A \rho) \quad ; \quad \]
\[\var(d'(A,X)) = O(\Vert A \Vert^2_F). \]
Finally, the deterministic procedure $d(\cdot)$ uses $T$ values $d'(A,x_1), d'(A,x_2), \ldots, d'(A,x_T)$, where each $x_i \leftarrow \M^{STAB}(\rho)$, to estimate $\Tr(A\rho)$ using the standard median-of-means approach. Using Chebyshev's inequality and the Chernoff bound, we obtain
$ \Pr_{s \leftarrow S} (\vert d(A,s) - \Tr(A \rho) \vert \leq \eps ) \geq 1 - \delta.$ \qedhere
\end{proof}

\section{Product distribution proof}\label{sec3}
Here we restate Theorem~\ref{intro:thm1} and provide its proof.
\begin{thm}
 Let $f: \X \times \Y \rightarrow  \Z \cup \{\bot\}$ be a partial function and $\mu$ be a product distribution supported on $f^{-1}(\Z)$. Denote $d=|\Z|$. Let $\eta >0$ and $0 \leq \eps \leq 1-1/d$. Then,
\begin{align*}
&\sfR^{1,\mu}_{2\eps-d\eps^2/(d-1)+\eta}(f) \\ & \quad \quad \leq  2\sfQ^{1,\mu,*}_{\eps}(f) +  O\L(\log\log (1/\eta)\R), \\
&\sfR^{1,\mu}_{2\eps-d\eps^2/(d-1)+\eta}(f) \\&\quad \quad \leq\sfQ^{1,\mu}_{\eps}(f) +  O\L(\log\log (1/\eta)\R),\\
&\sfR^{1,\mu_X}_{2\eps-d\eps^2/(d-1)+\eta}(f)\\ &\quad \quad \leq2\sfQ^{1,\mu_X,*}_{\eps}(f)  +  O\L(\log\log (1/\eta)\R),\\
&\sfR^{1,\mu_X}_{2\eps-d\eps^2/(d-1)+\eta}(f) \\&\quad \quad \leq \sfQ^{1,\mu_X}_{\eps}(f)  +  O\L(\log\log (1/\eta)\R).
\end{align*}
\end{thm}

\begin{proof}
The proofs of all the inequalities are all very similar. We give a detailed proof of the first and state the differences to obtain the other inequalities at the end. Recall we use the notation, $\psi$ to represent the state and also the density matrix $\ketbra{\psi}$, associated with $\ket{\psi}$. 

Let $S^y_z= \{x|f(x,y)=z\}$ and $Q^{1, \mu, *}_{\eps}(f) = a$. Consider an optimal distributional entanglement-assisted quantum communication strategy. Let the initial state be 
$$ \rho'_{XYAB}=\sum_{x,y} \mu(x,y) \ketbra{xy} \otimes \ketbra{\rho'_{AB}},$$
where $\ket{\rho'_{AB}}$ is the shared entanglement between Alice and Bob (Alice, Bob hold registers $A$, $B$ respectively). Alice applies a unitary $U:\cH_{XA} \rightarrow \cH_{XA'D}$ such that $U = \sum_{x} \ketbra{x} \otimes U^x$ (where $U^x : \cH_{A} \rightarrow \cH_{A'D}$ is a unitary conditioned on $X=x$) and sends across register $D$ to Bob. Let the state at this point be 
$$\rho_{XYA'Q}=\sum_{x,y} \mu(x,y) \ketbra{xy} \otimes \rho^x_{A'Q},$$
where $Q\equiv DB $. Since, $\rho_{XB} =\rho'_{XB}=\rho'_{X} \otimes \rho'_{B}=\rho_{X} \otimes \rho_{B}$,  from Fact~\ref{fact:boundnew} we have,
 \begin{align} \label{eq:X:Q}
    \mI(X:Q)_\rho  \leq  2\log (|D|) = 2a.
 \end{align}Bob performs measurement $\M^y$ with POVM elements $\{E^y_z : \forall z \in \Z \}$ on register $Q$ conditioned on $Y=y$ to output $f(x,y)$. Then,
	$$ \sum_{x,y} \mu(x,y) \trp{\rho^x_{Q} E^y_{f(x,y)}} \geq 1- \eps .$$This implies, 
\begin{align} \label{eq:prob_conv}
  1-\eps &\leq   \sum_{x,y} \mu(x,y) \trp{\rho^x_Q E^y_{f(x,y)}} \nn \\
    &=  \sum_y \mu(y)  \trp{\sum_{x} \mu(x) \rho^x_Q E^y_{f(x,y)}} \nn \\
    &=\sum_y \mu(y)\sum_{z\in\Z} \trp{ \sum_{x\in S^y_{z}} \mu(x) \rho^x_Q  E^y_{z}} \nn \\ 
    &=\sum_y \mu(y)\sum_{z\in\Z}\mu^y(z) \trp{ \rho_Q^{y,z} E^y_{z}}, 
\end{align}
where we defined $\mu^y(z)\defeq \sum_{x\in S^y_{z}} \mu(x)$ and $\rho_Q^{y,z}\defeq \frac{1}{\mu^y(z)}\sum_{x\in S^y_{z}} \mu(x) \rho^x_Q$. Note that $\rho_Q^{y,z}$ are density matrices and $\sum_{z\in\Z}\mu^y(z)=\sum_{z\in\Z}\sum_{x\in S^y_{z}}\mu(x)=\sum_{x\in\X}\mu(x)=1$. We can view $\sum_{z\in\Z}\mu^y(z) \trp{ \rho_Q^{y,z} E^y_{z}}$ as the success probability of distinguishing the cq-state $\rho^y_{ZQ}=\sum_{z\in \Z}\mu^y(z)  \proj{z} \otimes \rho^{y,z}_Q$ with measurement $\M^y$ with POVM elements $\{E^y_z: \forall z \in \Z \}$. We have,\
\begin{align}
    1-\eps \leq \sum_y \mu(y) \Pr[Z^y=\M^y(Q)], \nn
\end{align}
where we defined the random variable $Z^y\defeq f(X,y)$. Applying the function $g$ of Fact~\ref{thm:pgm-opt} to both sides and using the convexity of $g$, we have
\begin{align} 
   g(1-\eps) &\leq  g\L(\sum_y \mu(y) \Pr[Z^y=\M^y(Q)]\R) 
   \nn \\ 
   &\leq \sum_y  \mu(y) g\L(\Pr[Z^y=\M^y(Q)]\R). \label{eq:post-caushy}
\end{align}

We now fix $y$ and replace the optimal measurement Bob does by the PGM $\M^{pgm,y}_Z$. $\M^{pgm,y}_Z$ consists of POVM elements $E^{pgm,y}_z=A^{-1/2} A^y_z A^{-1/2}$ for all $z\in\Z$, where $A^y_z= \mu^y(z)\rho^{y,z}_Q$, and $A=\sum_{z\in\Z}A^y_z$. Note that $A^y_z=\sum_{x\in S^y_z} \mu(x) \rho^x_Q$, and $A=\sum_x \mu(x) \rho^x_Q$ is independent of $y$. From Fact~\ref{thm:pgm-opt}, the optimality of PGM, we have that for all $y$,  
$$g\L(\Pr[Z^y=\M^y(Q)]\R) \leq \Pr[Z^y=\M^{pgm,y}_Z(Q)].$$ 
Using Equation~\eqref{eq:post-caushy},
\begin{align}\label{eq:before_split}
    g(1-\eps) \leq &\sum_y \mu(y)  \Pr[Z^y=\M^{pgm,y}_Z(Q)] \nn \\
    =& \sum_y \mu(y) \trp{\sum_x \mu(x) \rho^x_Q E^{pgm,y}_{f(x,y)}},
\end{align}
where the last line follows logic similar to Equation~\eqref{eq:prob_conv}.
Notice that   
\begin{align}
    E^{pgm,y}_z &= A^{-1/2} A^y_z A^{-1/2} \nn \\&=A^{-1/2}\left(\sum_{x\in S^y_z} \mu(x) \rho^x_Q\right) A^{-1/2} \nn \\
    &= \sum_{x\in S^y_z} A^{-1/2} A_x A^{-1/2} \nn \\&= \sum_{x\in S^y_z} E^{pgm}_x=\sum_{x'} \delta_{z,f(x',y)}E^{pgm}_{x'}, \nn
\end{align}
where $\{E^{pgm}_x: \forall x \in \X \}$ are the POVM elements of the PGM that guesses $X$ from $Q$. 
 Therefore, instead of doing $ \M^{pgm,y}_Z$, we can measure with $\M^{pgm}_X$, getting a guess $x'$, and then compute $f(x',y)$ as our guess of $Z^y=f(X,y)$. That is, 
 \begin{align} \label{eq:split}
&\trp{\sum_x \mu(x) \rho^x_Q E^{pgm,y}_{f(x,y)}} \nn \\&     = \trp{\sum_x \mu(x) \rho^x_Q \sum_{c\in S^y_{f(x,y)}}E^{pgm}_{c}}.
\end{align}



More precisely, define $C \equiv \M^{pgm}_{X}(Q)$. Since $C$ is a classical random variable that is independent of $y$, Alice can compute $C$ by herself. Consider the intermediate classical one-way communication protocol where Alice computes and sends $C = \M^{pgm}_{X}(Q)$ to Bob, and Bob predicts $f(x,y)$ with $z=f(c,y)$. The success probability of this intermediate protocol is

\begin{align}
     &\sum_y \mu(y) \sum_x \mu(x) \sum_{c \in S^y_z} \Pr[C=c|X=x]  \nn \\
    =& \sum_y \mu(y) \trp{\sum_x \mu(x) \rho^x_Q \sum_{c\in S^y_{f(x,y)}}E^{pgm}_{c}} \nn \\
    &\geq g(1-\eps), 
\end{align}
where we used Equation~\eqref{eq:split} and Equation~\eqref{eq:before_split} in the last line.

In this intermediate protocol, the message $C$ that Alice sent is not short. In fact, it has the same length as $X$. However, $C$ has low max-information with $X$. By Equation~\eqref{eq:X:Q} and Fact~\ref{fact:mono} we have  
\begin{align}
    &\mI(X:C)=\mI(X : \M^{pgm}_{X}(Q)) \nn \\& \leq \mI(X:Q)_\rho \leq  2a. \nn
\end{align}
Therefore using Fact~\ref{fact:compress}, we can compress $C$ and get a classical one-way communication protocol with message $M$ and public-coin $R$, such that 
\begin{align}
   & \ell(M) \leq  2a +  O(\log\log (1/\eta)),  
\end{align}
and success probability $g(1-\eps)-\eta$. Therefore $\sfR^{1,\mu}_{2\eps-d\eps^2/(d-1)+\eta}(f) \leq 2a + O(\log\log(1/\eta))$ which gives the desired (bound). 

To obtain the second inequality we note that Alice and Bob did not have the starting state ${\rho'_{AB}}$,~so the register $B$ is empty, and in Equation~\eqref{eq:X:Q} we have instead 
 \begin{align}
     \mI(X:Q)_\rho= \mI(X:D)_\rho  \leq  \log (|D|) =   a.
 \end{align}
The inequality is because of Fact~\ref{fact:mono}. Everything else follows. 
 
To obtain the third inequality we basically condition on every possible $y$, i.e. changing "$\sum_y \mu(y)$" in the proof to "for all $y$" and  "$\sum_{x,y} \mu(x,y)$" to "for all $y$ , $\sum_x \mu(x)$". We also do not need to use the convexity of Fact~\ref{thm:pgm-opt}.

To obtain the fourth inequality, we combine the changes to obtain the second and the third inequalities.~\qedhere 
\end{proof}


\section{Non-product distribution proof}\label{sec4}

\suppress{We upper bound  distributional classical one-way communication complexity by  distributional quantum one-way communication complexity, using ideas of~\cite{HK19}. Given a quantum message with good winning probability for quantum distributional one-way communication, we use the idea of~\cite{HK19} to show that there exists a classical message that achieves similar winning probability for  distributional classical one-way  communication.}
Here we restate Theorem~\ref{intro:thm2} and provide a proof.
\begin{thm}
Let $\eps, \eta >0$ such that $\eps/\eta + \eta < 0.5$. Let $f: \X \times \Y \rightarrow \{0,1,\bot \}$ be a partial function and $\mu$ be a distribution supported on $f^{-1}(0) \cup f^{-1}(1)$. Then,
	$\sfD^{1,\mu}_{3\eta}(f)=O\left(  \frac{\sfcs(f) }{\eta^3}  \sfQ^{1,\mu}_{{\eps}} (f)  \right)$, where
      	\[\sfcs(f) = \max_{y} \min_{z\in\{0,1\}} \{ \vert \{x~|~f(x,y)=z\} \vert\}  \enspace.\] 
\end{thm}
\begin{proof}
Let $S^y_0= \{x|f(x,y)=0\}$, $S^y_1= \{x|f(x,y)=1\}$ and $\sfQ^{1, \mu}_{\eps}(f) = a$. Consider an optimal quantum protocol $\cP$ where Alice prepares the quantum message $Q$ according to the cq-state $\psi_{XQ}=\sum_x \mu(x) \proj{x}\otimes \psi^x_Q$, and Bob performs measurement $\M^y$ with  POVM elements $\{E^y_0,E^y_1\}$ to output $f(x,y)$. We can assume, Alice sends $\psi^x_Q$ along with its canonical purification since it only increases the quantum communication by a multiplicative factor of $2$. Also, we can assume that for every $y$, POVM elements $\{E^y_0,E^y_1\}$ are projectors, since Alice can send ancilla and Bob can realize POVM operators as projectors (Fact~\ref{fact:naimark}). From here on, we assume $\psi^x_Q$ is a pure state $\ket{\psi^x_Q}$ and $E^y=\{E^y_0,E^y_1\}$ are projectors for every $x,y$ respectively. Since $\sfQ^{1,\mu}_{\eps}(f) = a$, we have $\log|Q| \leq a$ and
\begin{align*}
    \sum_{x,y} \mu(x,y) \trp{\ketbra{\psi^x_Q} E^y_{f(x,y)}} \geq 1- \eps .
\end{align*}


  For all $y$, define,
  \begin{align*}
      & \forall x \in S^y_0, \quad  \ket{ \tilde{\psi}^x_{Q}}  \defeq {E^y_0 \ket{\psi^x_{Q}}} \quad \\
      & \text{and}  \quad \tilde{E}^y_0 \defeq  \mathsf{Proj}\left(\supp\left(\sum_{x \in S^y_0 } \ketbra{\tilde{\psi}^x_{Q}}\right)\right)   ,
  \end{align*}
\begin{align*}
      & \forall x \in S^y_1, \quad  \ket{\tilde{\psi}^x_{Q}}  \defeq {E^y_1 \ket{\psi^x_{Q}}} \quad  \\
      & \text{and}  \quad \tilde{E}^y_1 \defeq \mathsf{Proj}\left(\supp\left(\sum_{x \in S^y_1 }\ketbra{\tilde{\psi}^x_{Q}}\right) \right), 
  \end{align*}
where $ \ket{\tilde{\psi}^x_{Q}} $ are unnormalized vectors with length less than $1$. Note that $\tilde{E}^y_0 \le E^y_0, \tilde{E}^y_1 \le E^y_1, $ $\trp{\ketbra{\psi^x_{Q}} E^y_{0}} =\trp{\ketbra{\psi^x_{Q}} \tilde{E}^y_{0}}$ for every $x \in S^y_0$ and $\trp{\ketbra{\psi^x_{Q}} E^y_{1}} =\trp{\ketbra{\psi^x_{Q}} \tilde{E}^y_{1}}$ for every $x \in S^y_1$. Also, $\Vert \tilde{E}^y_{0} \Vert_F^2 \leq \vert S^y_0 \vert $ and $\Vert \tilde{E}^y_{1} \Vert_F^2 \leq \vert S^y_1 \vert$ from Fact~\ref{factprojfrobrank}. Let 
\begin{align*}
    K  = \max_{y} \min \{ \Vert \tilde{E}^{y}_0 \Vert^2_F, \Vert \tilde{E}^{y}_1 \Vert^2_F\}.
\end{align*}
Let $b_y =i $ be such that $ \vert S^y_i \vert \leq \vert S^y_{1-i} \vert$. For $x$ such that $f(x,y) = b_y$, we have \begin{align}\label{eq100}
   \trp{\ketbra{\psi^x_Q} \tilde{E}^y_{b_y}}  =\trp{\ketbra{\psi^x_Q} E^y_{b_y}} ,
\end{align}and if $f(x,y) = 1-b_y$, then  \begin{align}\label{eq101}
   \trp{\ketbra{\psi^x_Q} \tilde{E}^y_{b_y}}  \leq \trp{\ketbra{\psi^x_Q} E^y_{b_y}}.
\end{align} 
	 
The (intermediate) classical protocol $\cP_1$ is as follows. 	 
	\begin{enumerate}
	    \item Alice (on input $x$) prepares $T = O\left(\frac{K}{\eta^2} \log \left( \frac{1}{\eta} \right) \right)$ copies of  $\ket{\psi^x_Q}$, i.e. $\ket{\psi^x_Q}^{\otimes T}$ and measures them independently in stabilizer measurement ($\M^{STAB}$) to generate a classical random variable $S_x= \M^{STAB}(Q)^{\otimes T}$. 
	    \item Alice sends $S_x$ to Bob. Note that $\ell(S_x) = O(T \ell(Q)^2) =O(Ta^2)$. 
	    \item Bob (on input $y$) estimates  $\trp{\ketbra{\psi^x_Q} \tilde{E}^{y}_{b_y}}$ via a deterministic procedure $d(.)$ such that (from Fact~\ref{hk_trace}) 
	 \begin{align} \label{eq:stabsample}
	     &\Pr_{s \leftarrow S_x} ( \vert d(\tilde{E}^y_{b_y},s) - \trp{\ketbra{\psi^x_Q} \tilde{E}^y_{b_y}} \vert \leq \eta ) \nn \\  & \geq 1- \eta.
	 \end{align}
	 \item If Bob's estimated value turns out to be less than $0.5$, he outputs $1-b_y$, otherwise $b_y$.
	\end{enumerate} 
	 

Let $ \mathcal{I}(x,y,s) $ be the indicator function such that $\mathcal{I}(x,y,s)=1 $ if subsample $s$ results in Bob (with input $y$) estimating $\trp{\ketbra{\psi^x_Q} \tilde{E}^{y}_{b_y}}$ upto additive error $\eta$. 
For every $x,y$, define $\mathsf{good}_{xy} \defeq \{s\in\supp(S_x) ~|~  \mathcal{I}(x,y,s)=1 \}$ and $\mathsf{bad}_{xy} = \supp(S_x) \setminus \mathsf{good}_{xy}$. 
Define,  $\mathsf{good} \defeq \{(x,y)~|~ \err_{x,y}(\cP,f) \leq \eps/\eta\}$. From Markov's inequality $\Pr_{(x,y) \leftarrow \mu}((x,y) \in \mathsf{good}) \geq 1 - \eta$.  Using Equations~\eqref{eq:stabsample},~\eqref{eq100},~\eqref{eq101} and $\eps/\eta + \eta <0.5$, we note that when $(x,y) \in \mathsf{good}$ and $s \in \mathsf{good}_{xy}$,  Bob gives correct answer for $f(x,y)$. Thus, the probability of correctness of $\cP_1$ is at least, 
 \begin{align*}
  \sum_{(x,y)\in \mathsf{good}} \mu(x,y) \cdot  \Pr(S_x \in \mathsf{good}_{xy}) \geq 1-2\eta.
\end{align*}	
In $\cP_1$, the message $S$ (averaged over $x$) that Alice sent is of size $O(Ta^2)$. However, $S$ has low mutual information with $X$. Using Fact~\ref{fact:mono}, we have  
\begin{align*}
    &\imax(X:S)=\imax(X : (\M^{STAB}(Q))^{\otimes T}) \\&\leq  \imax(X:Q^{\otimes T})_\psi \leq Ta.
\end{align*}

Therefore using Fact~\ref{fact:compress}, we can compress $S$ and get a classical one-way communication protocol with message $M$ and public-coin $R$, such that 
\begin{align}
   & \ell(M) \leq  Ta +  O(\log\log (1/\eta)),  
\end{align}
and success probability $1-2\eta-\eta =1-3\eta$. Thus, we have $\sfR^{1,\mu}_{3\eta}(f) \leq Ta +  O(\log\log (1/\eta))$. Noting $K \leq \sfcs(f)$, we have the desired. \qedhere

\section{Acknowledgment}\label{sec:acknowledgement}
The work of RJ was supported in part by the National Research Foundation, Singapore, under Grant NRF2021-QEP2-02-P05; and in part by the Ministry of Education, Singapore, through the Research Centers of Excellence Program.
The work of NGB was done while he was a graduate student at the Centre for Quantum Technologies and was supported by the Ministry of Education, Singapore through the Research Centers of Excellence Program. The work of HL was funded by MOST Grant no. 110-2222-E-007-002-MY3.

	\end{proof}

\bibliographystyle{plainnat}
\bibliography{KM-ref}

\end{document}